\documentclass[a4paper,11pt]{bxjsarticle}

\usepackage[japanese,english]{babel}

\usepackage{amsmath,amsthm,amssymb,amsfonts}
\usepackage{mathtools}
\usepackage{bm}

\usepackage{color}
\usepackage{colortbl}
\usepackage{lmodern}
\usepackage{textcomp}
\usepackage{type1cm}
\definecolor{light-gray}{rgb}{0.83, 0.83, 0.83}

\usepackage{ascmac}
\usepackage{fancybox}

\usepackage{booktabs}
\usepackage{multirow}
\usepackage{makecell}

\usepackage{graphics}
\usepackage{graphicx,xcolor}
\usepackage{here}
\usepackage{float}
\usepackage[width=.75\textwidth]{caption}
\usepackage{enumerate}

\usepackage{titlesec}
\titleformat*{\section}{\Large}
\titleformat*{\subsection}{\large}

\usepackage[hang,flushmargin]{footmisc}

\usepackage{hyperref}
\usepackage{etoolbox}
\patchcmd{\thebibliography}{\section*{\refname}}{}{}{}

\usepackage{apacite}
\AtBeginDocument{%
}

\newtheorem{assumption}{Assumption}[]
\newtheorem{proposition}{Proposition}[]

\renewcommand{\thetable}{ \arabic{table}}
\renewcommand{\thefigure}{ \arabic{figure}}
 
\DeclareMathOperator*{\argmin}{arg\,min}

\usepackage{lscape}
\usepackage{rotating}

\begin{document}
\selectlanguage{english} 

\title{MSE-Optimal Difference-in-Differences Estimator\thanks{
  I sincerely appreciate Chishio Furukawa for his generous guidance and advice. 
  I also thank Shosei Sakaguchi for insightful comments.
  All errors are my own.
  The replication package for this paper is available at \url{https://github.com/yamato5810/MSE-Optimal_Difference-in-Differences_Estimator}
}}
\author{Yamato Igarashi\thanks{
  Graduate School of Economics, The University of Tokyo. igarashi-yamato-by@g.ecc.u-tokyo.ac.jp.
}}
\date{\today}

\maketitle

\vspace{2.0cm}

\begin{abstract}
  This paper develops a difference-in-differences (DiD) estimation method that selects the optimal length of pre-trends by minimizing the mean squared error (MSE).
  Conventional DiD regression models, such as the two-way fixed effects model or the event study model, may suffer from accuracy and validity concerns. 
  If the sample size is small, the estimator may have a larger variance.
  Also, pre-tests often lack power to detect violations of the parallel trends assumption as \citeA{Roth:2022} highlights.
  By focusing on the bias and variance tradeoff, the proposed method derives the MSE-optimal estimator from the optimal length of pre-trends.
  Simulation results and an empirical application demonstrate the practical applicability of the proposed method. 
\end{abstract}

\vspace{1cm}

\textbf{\textit{Keywords:}} Difference-in-Differences, Pre-trends, Mean Squared Error

\newpage

\newpage

\section{Introduction}
Difference-in-differences (DiD) is a widely used method in economics for estimating causal effects.
One of the well-known studies that implemented DiD is \citeA{Card_and_Krueger:2022}, which studies the effect of the minimum wage increase in New Jersey.
Recently, there have been many theoretical advances in DiD methods enabling researchers to estimate treatment effect in a wide range of settings. 

In practice, DiD is usually implemented via two regression models: the two-way fixed effects (TWFE) model and the event study model.
Aside from the problem of staggered treatment timing, those models have several problems.
One of the problems is that precise estimation is difficult when the sample size is small.
Event study plots, which are commonly used to present event study estimates, tend to have wider confidence intervals particularly when the number of observations is small.
This can be viewed as a variance problem.
Another problem is that the pre-trends test is not always effective \cite{Roth:2022}.
Even though there may be a potential bias between the treatment and control groups, the pre-trends test cannot detect it.
This can be viewed as a bias problem. 

To address these problems, this paper proposes an MSE-optimal approach to difference-in-differences estimation.
The framework is the following.
When longer pre-trends are included, the estimator tends to have a smaller variance.
At the same time, the estimator may be biased if there is a potential bias between the treatment and control groups.
On the other hand, the estimator tends to have a larger variance when only a shorter pre-trends are included.
Also, the estimator is likely to be less biased.
Using the bias and variance tradeoff, this method derives the MSE-optimal estimator by choosing the optimal length of the pre-trends. 
I illustrate this method through simulations and an empirical application.
The results show that the proposed method performs well in terms of MSE relative to conventional methods.

Several studies examine pre-trends in the DiD setting.
\citeA{Roth:2022} shows that the conventional pre-trends tests often have low power and that depending on such tests can distort estimation and inference.
\citeA{Rambachan_and_Roth:2023} propose robust inference procedures that allow for possible violations of the parallel trends assumption.
My paper takes the possibility of such bias seriously and proposes a procedure that minimizes the MSE of the estimator by explicitly balancing the bias and variance.
\citeA{Egami_and_Yamauchi:2023} show how multiple pre-treatment periods can improve DID estimation and propose double DID, which can be more efficient and rely on weaker assumptions than the conventional DID.
My paper also focuses on the role of pre-treatment periods in DiD estimation, by viewing the length of the pre-trends as a choice variable in estimation.
\citeA{Ishimaru:2026} studies treatment effect estimation in panel data without relying on the parallel trends assumption and proposes alternative conditions for identifying the average treatment effect on the treated.
By contrast, my paper works within a DiD framework in which deviations from parallel trends are assumed to be limited, so that the resulting distortion can be interpreted as a bias and incorporated into the MSE criterion.

The bias-variance tradeoff is also fundamental in the selection of the optimal bandwidth in regression discontinuity design (RDD) \cite{Imbens_and_Kalyanaraman:2012, Calonico_and_Cattaneo_and_Titiunik:2014}. 
A wider bandwidth includes more observations in estimation, which generally lowers variance. 
At the same time, however, observations farther from the cutoff may be less informative about the local treatment effect at the threshold, so a wider bandwidth may induce greater bias. 
Conversely, a narrower bandwidth focuses on observations closer to the cutoff, thereby reducing potential bias because the identifying assumptions are more plausible.
However, fewer observations are available; therefore the estimator typically has higher variance. 
Thus, the bandwidth choice in RDD is naturally characterized as a bias-variance tradeoff. 
This paper applies a similar idea to the DiD setting: extending the pre-treatment period may improve precision, but it may also increase bias because the conventional pre-trends test may fail to detect the bias.

The rest of the paper is organized as follows: In the next section, 
I describe the canonical DiD models as baseline DiD models.
In Section $3$, I show the proposed model and explain the implementations in the static and dynamic treatment effect cases.
Then, I illustrate the model with simulations in Section $4$.
I also apply the model to the empirical example in Section $5$, and I conclude in Section $6$.

\section{Baseline}
This section describes the standard DiD framework as compared with the model this paper proposes. 
Section $2.1$ is about the basic $2 \times 2$ design and some key assumptions.
Section $2.2$ explains the two-way fixed effects model which is a commonly used regression model for static treatment effect.
Section $2.3$ illustrates the event study model which is a commonly used regression model for dynamic treatment effect.

\subsection{Canonical Difference-in-differences}
The simplest DiD setting is a $2 \times 2$ (2 groups (treatment group vs control group) and 2 periods (pre-treatment vs post-treatment)) design. 
There have been many advances related to DiD methods even in recent years (see, e.g. \citeA{Baker_etal:2022}, \citeA{Baker_etal:2026} and \citeA{Roth_etal:2023} for a review)
and those newer methods differ technically from one another, but they fundamentally depend on the same idea.
Therefore, I begin with the simplest $2 \times 2$ design 
and specify some important assumptions. 

The notation used in this paper is as follows.
$Y_{i,t}$ denotes an outcome of individual $i$ and time $t$.
$D_{i,t}$ indicates whether or not the individual $i$ receives treatment in period $t$.
$Y_{i,t}(D_{i,t} = 1) = Y_{i,t}(1)$ represents the potential outcome of $Y$ if the individual $i$ receives a treatment at period $t$, and
$Y_{i,t}(D_{i,t} = 0) = Y_{i,t}(0)$ represents the potential outcome of $Y$ if the individual $i$ does not receive a treatment at period $t$.
Lastly, let $t^*_i$ be the time just before the individual $i$ receives a treatment.
Hereafter, $t^*_i$ is set to $-1$ unless specified otherwise.

One common target parameter of DiD is the average treatment effect on the treated (ATT). 
In the $2 \times 2$ design, the ATT can be expressed as
\begin{align*}
  \tau \equiv \mathbb{E}[Y_{i,1}(1) - Y_{i,1}(0)|D_{i} = 1]
\end{align*}

To obtain this ATT, the following three assumptions - Assumption~\ref{assumption:1}: Parallel Trends, Assumption~\ref{assumption:2}: No Anticipation and Assumption~\ref{assumption:3}: No Spillover - are crucial.

\begin{assumption}\label{assumption:1}
  {\normalfont Parallel Trends}\\
  \begin{equation*}
    \mathbb{E}[Y_{i,t}(0) - Y_{i,t-1}(0)|D_i = 1] = \mathbb{E}[Y_{i,t}(0) - Y_{i,t-1}(0)|D_i = 0] 
  \end{equation*}
\end{assumption}
\vspace{0.5cm}

\noindent
Assumption~\ref{assumption:1} states the average changes in outcomes are the same for the treatment group and the control group if both groups do not receive the treatment. 

\begin{assumption}\label{assumption:2}
  {\normalfont No Anticipation}\\
  \begin{equation*}
    Y_{i,t} = Y_{i,t}(0) \ \textnormal{for all} \ t \leq t^{*} \ \textnormal{and} \ i \ \textnormal{with} \ D_i = 1
  \end{equation*}
\end{assumption}
\vspace{0.5cm}

\noindent
Assumption~\ref{assumption:2} is an assumption that 
implies the treatment does not affect any actions before the treatment timing.
This assumption can be interpreted as a situation that no one could know whether or not they would receive a treatment.

\begin{assumption}\label{assumption:3}
  {\normalfont No Spillover\footnote{Although this paper maintains the No Spillover assumption, recent work has developed DiD methods that explicitly account for spillover effects; see \citeA{Butts:2023} and \citeA{Fiorini_Lee_and_Pfeifer:2024}.}}\\
  \begin{equation*}
   \forall j \neq i, \ Y_{i,t} (D_i, D_j) =  Y_{i,t} (D_i, D_j') \ \textnormal{for all} \ D_j\ \textnormal{and} \ D_j' 
  \end{equation*}
  {\normalfont }
\end{assumption}

\vspace{0.5cm}
\noindent
Assumption~\ref{assumption:3} states an individual's outcome is not affected by the treatment status of others.  

Under these assumptions, ATT in the $2 \times 2$ setting can be transformed into the following equation.

\begin{align}\label{ATT}
  \tau \ =& \ \mathbb{E}[Y_{i,1}(1) - Y_{i,1}(0)|D_{i} = 1]	\nonumber \\ 
  =& \ \mathbb{E}[Y_{i,1}(1) - Y_{i,0}(0)|D_{i} = 1] -\mathbb{E}[Y_{i,1}(0) - Y_{i,0}(0)|D_{i} = 0] \\
  =& \ \mathbb{E}[Y_{i,1} - Y_{i,0}|D_{i} = 1] -\mathbb{E}[Y_{i,1} - Y_{i,0}|D_{i} = 0] \nonumber
\end{align}

The above discussion presents the $2 \times 2$ DiD framework.
I extend the $2 \times 2$ DiD framework to a more general setting with multiple individuals observed over multiple time periods, focusing in particular on the extension from two periods to many periods.
Throughout the analysis, I maintain Assumptions 2 and 3.

\subsection{Two-Way Fixed Effects}
In practice, researchers typically estimate the ATT in a DiD framework using regression-based methods. 
One of the most commonly used specifications is the two-way fixed effects (TWFE) model. 
The TWFE specification without covariates is given by:

\begin{equation}\label{TWFE}
  Y_{it} =  \mu_i + \theta_t + \beta D_{it} + \varepsilon_{it}
\end{equation}

\noindent
where $\mu_i$ is an individual fixed effect and $\theta_t$ is a time fixed effect.
Then, the coefficient $\beta$ represents the ATT.
This model is suitable when the treatment effect is ``static'' (e.g. Figure\ref{fig: canonical plots} A). 
If the treatment effect is constant across post-treatment periods, a time-invariant $\beta$ is appropriate.
Although the TWFE model had been used widely in empirical research, this model has a problem with staggered treatment timing \cite{Goodman-Bacon:2021,Callaway_and_Sant'Anna:2022,de_Chaisemartin_and_d'Haultfoeuille:2020,de_Chaisemartin_and_d'Haultfoeuille:2023}.
This paper assumes a treatment occurs at the same time for all individuals, but it might be possible to extend to staggered treatment timing.

\subsection{Event Study}
Another model to derive a DiD estimator is the event study model.
This model is suitable when the treatment effect is ``dynamic'' (e.g. Figure\ref{fig: canonical plots} B). 
The event study model without covariates can be written as:

\begin{equation}\label{simple event study}
  Y_{it} =  \mu_i + \theta_t + \sum_{l \neq -1} \beta_l 1\{t - t^*_i = l\} + \varepsilon_{it}
\end{equation}

\noindent
where $\mu_i$ and $\theta_t$ are the same as the TWFE model.
The difference is the third term, which makes it possible to estimate the dynamic treatment effect.
When the treatment effect is dynamic, $\beta_l$ must be time-variant to capture the change in the treatment effect.
This paper does not address this issue, but staggered treatment causes a problem in the event study model similarly to the TWFE model \cite{Sun_and_Abraham:2021}.

\section{Model}
This section describes the model this paper proposes.
I first clarify the idea of the proposed model in Section $3.1$.
Then, I illustrate two cases depending on the type of treatment effect (static and dynamic)
in Section $3.2$ and $3.3$.

\subsection{Model Overview}
The main idea of this paper is to construct the MSE-optimal DiD estimator by changing the length of pre-trends.
I first explain why it is appropriate to consider the MSE-minimizing approach in this context.
In general, estimators including DiD estimators have smaller standard errors when the sample size is large.
This means a DiD estimator would be more precise if the pre-trends are longer. 
At first glance, longer pre-trends may seem preferable to shorter ones.
However, longer pre-trends might bias the estimator unintentionally.
In DiD studies, it is common to check the parallel trends assumption by examining whether or not pre-treatment estimates are statistically close to zero using $95\%$ confidence intervals (CI),
but such a pre-trends test may fail to detect bias \cite{Roth:2022}.
Therefore, it is possible that the estimator is biased. 
In particular, if there are longer potentially nonparallel pre-trends, the bias would be larger.
Hence, longer pre-trends might have a smaller variance and a larger bias, and shorter pre-trends might have a larger variance and a smaller bias. 
It can be said that there is a bias and variance tradeoff.
Note that MSE is expressed as the following:

\begin{equation}\label{MSE}
  \text{MSE} = \text{Bias}^2 + \text{Variance}
\end{equation}

\noindent
I define the optimal choice as the one that minimizes the MSE.

\begin{equation} \label{Min MSE}
  \min \text{MSE} = \min \big(\text{Bias}^2 + \text{Variance}\big)
\end{equation}

\noindent
This is why it is possible to incorporate the MSE-minimizing approach into the DiD framework by changing the length of pre-trends.

To analyze the proposed procedure more precisely, I consider an oracle benchmark.
Let $\mathcal{L}=\{0,1,\ldots,L\}$ denote the set of candidate length of pre-trends.
Also, let $\ell \in \mathcal{L}$ denote the length of pre-trends and $\ell_{\text{post}} (\geq 1)$ denote the length of post-treatment periods, which is fixed in this setting.
In this benchmark, I consider the TWFE specification in (\ref{TWFE}), and the target parameter is $\beta$.
For each $\ell \in \mathcal{L}$, let $\hat{\beta}_{\ell}$ denote the DiD estimator constructed with pre-treatment length $\ell$.
I impose the following two assumptions.
\begin{assumption}\label{assumption:oracle1}
  Assume that deviations from parallel trends in the pre-treatment periods follow a linear form, $\gamma \ell$.
\end{assumption}
\begin{assumption}\label{assumption:oracle2}
  Assume that the variance of the candidate estimators is finite.
\end{assumption}
\noindent
Then, the following proposition holds.

\begin{proposition}\label{prop:oracle}
  Suppose that the deviations from parallel trends in the pre-treatment periods satisfy Assumption \ref{assumption:oracle1} and the variance of the candidate estimators satisfies Assumption \ref{assumption:oracle2}.
  Also, suppose that the candidate set $\mathcal{L}=\{0,1,\ldots,L\}$ is finite.
  Then, there exists an MSE-optimal DiD estimator $\hat{\beta}^*$.
\end{proposition}

\begin{proof}
  For each $\ell \in \mathcal{L}$, the MSE of $\hat{\beta}_{\ell}$ is
  \begin{align*}
    \text{MSE}\big(\hat{\beta}_{\ell}\big) &= \text{Bias}\big(\hat{\beta}_{\ell}\big)^2 + \text{Var}\big(\hat{\beta}_{\ell}\big) \\
    &= \Big(\frac{\gamma \ell}{\ell + \ell_{\text{post}}}\Big)^2 + \text{Var}\big(\hat{\beta}_{\ell}\big) \ ( < \infty )
  \end{align*}
  Then, the optimal length of pre-trends $\ell^*$ is 
  \begin{align*}
    \ell^* = \argmin_{\ell \in \mathcal{L}}  \text{MSE}\big(\hat{\beta}_{\ell}\big) = \argmin_{\ell \in \mathcal{L}} \Bigg( \Big(\frac{\gamma \ell}{\ell + \ell_{\text{post}}}\Big)^2 + \text{Var}\big(\hat{\beta}_{\ell}\big) \Bigg)
  \end{align*}
  and the MSE-optimal DiD estimator $\hat{\beta}^*$ is 
  \begin{align*}
    \hat{\beta}^* = \hat{\beta}_{\ell^*}
  \end{align*}
\end{proof}

\noindent
This oracle rule is not feasible in practice because the true bias and variance are generally unknown. 
However, it serves as a theoretical benchmark for the feasible selection rule developed below.

In the feasible implementation, the variance is approximated using the squared standard error of the estimate.
On the other hand, the true bias is unobserved in practice.
I therefore use the estimator with pre-treatment length $0$ as a benchmark.
Specifically, I define the approximated bias for the estimator with pre-treatment length $\ell$ as
\begin{equation*}
  \text{bias}_{\text{pre-trends length} = \ell} \equiv \text{estimate}_{\text{pre-trends length} = \ell} - \text{estimate}_{\text{pre-trends length} = 0}
\end{equation*}

\begin{figure}[H]
  \centering
  \includegraphics[width=0.70\columnwidth]{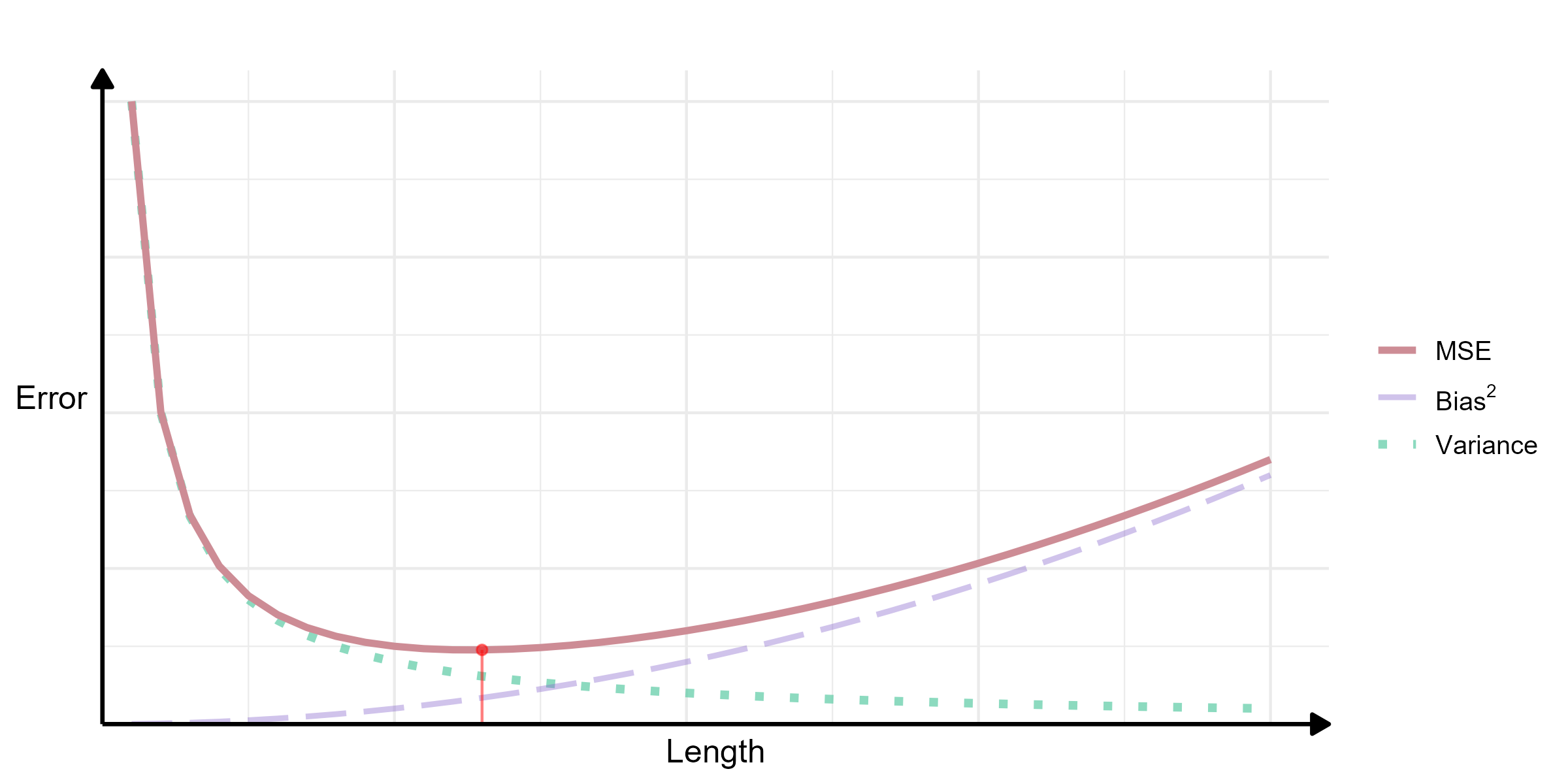}
  \caption{Bias-variance tradeoff}
  \label{fig: Bias-variance tradeoff}
  \begin{minipage}{14cm}
    \footnotesize
    Notes: This figure explains the image of the bias and variance tradeoff generated by changing the length of pre-trends.
    The variance is larger when the pre-treatment period is shorter.
    On the other hand, the bias is smaller when the pre-treatment period is shorter.
    An opposite relationship holds when the pre-treatment period is longer.
    Then the optimal length of pre-trends would be obtained when the MSE ($= \text{Bias}^2 + \text{Variance}$) is minimized.
  \end{minipage}
\end{figure}

\subsection{Static Treatment Effect Estimation}
I consider the MSE-minimizing approach in the case of a static treatment effect.
In this case, the TWFE model is often used as I mentioned in the previous section.
Therefore, I incorporate the MSE-minimizing approach into the TWFE model.
Figure\ref{fig: Changing pre-trends static version} shows the changes of the outcome with pre-trends length $9, 4,$ and $ 0$.
When the length of pre-trends is $9$, more observations are available to estimate $\beta$ in (\ref{TWFE}).
Hence, the variance of the estimate would be smaller.
However, pre-trends are not completely parallel and this would bias the estimate.
When the length of pre-trends is $0$, there are fewer observations to estimate $\beta$ in (\ref{TWFE}).
Thus, the variance of the estimate would be larger.
On the other hand, there are no pre-trends that would induce bias in the estimate.
When the length of pre-trends is $4$, there is a moderate number of observations to estimate $\beta$ in (\ref{TWFE}).
Namely, the variance of the estimate would be moderate.
Also, there is an intermediate length of pre-trends that would cause the estimate to be biased.
By calculating the variance and the bias of each length of pre-trends, the optimal length of pre-trends and the MSE-optimal TWFE estimator can be obtained.

\begin{figure}[H]
  \centering
  \includegraphics[width=0.90\columnwidth]{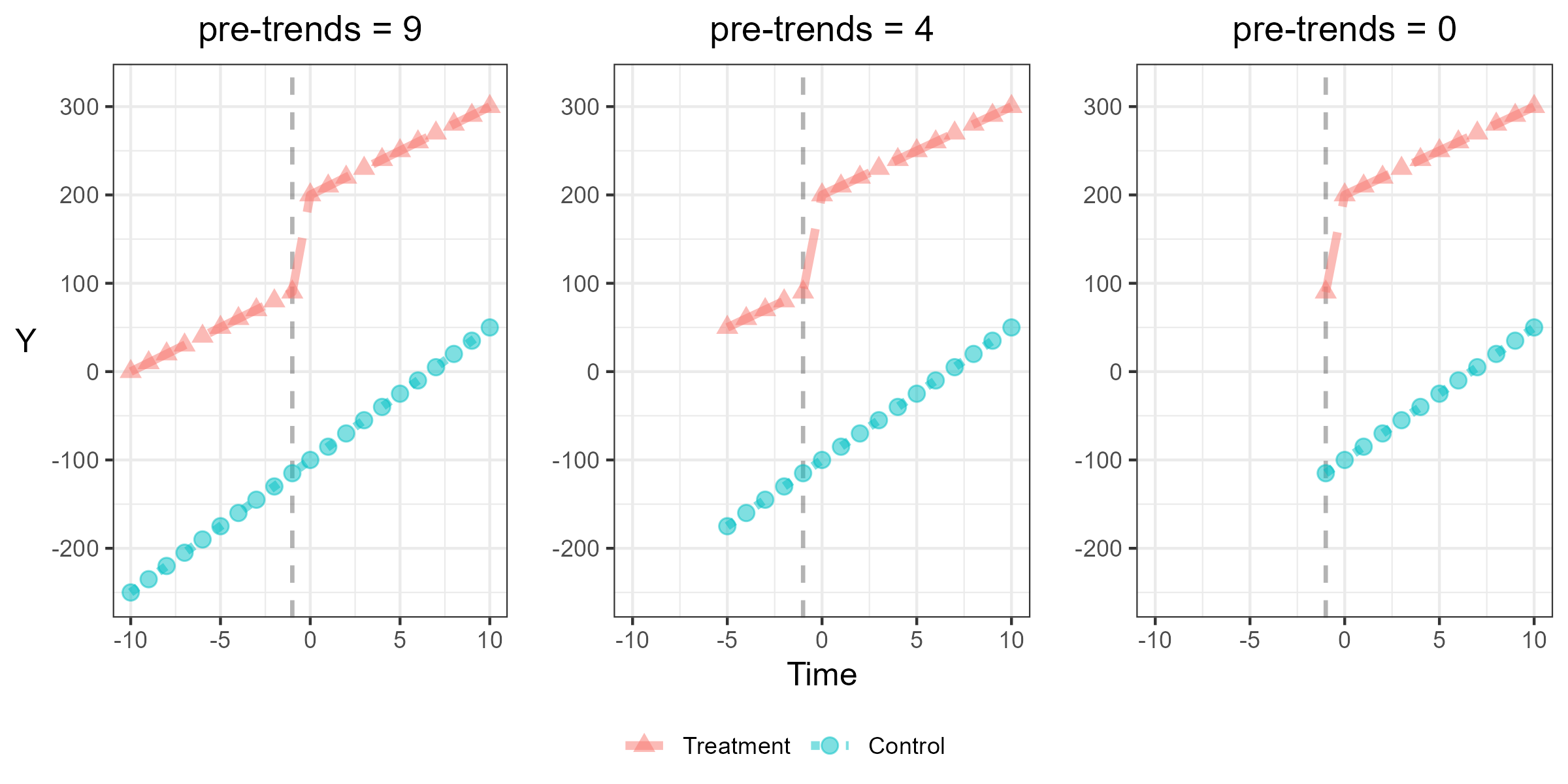}
  \caption{Different pre-trends length with static treatment effect}
  \label{fig: Changing pre-trends static version}
  \begin{minipage}{14cm}
    \footnotesize
    Notes: In the left figure, the length of the pre-trends is $9$.
    The estimate of the TWFE model with this length of pre-trends might have smaller variance and larger bias.
    In the middle figure, the length of the pre-trends is $4$.
    The estimate of the TWFE model with this length of pre-trends might have moderate variance and bias.
    In the right figure, the length of the pre-trends is $0$.
    The estimate from the TWFE model with this length of pre-trends might be larger variance and smaller bias.
  \end{minipage}
\end{figure}

\subsection{Dynamic Treatment Effect Estimation}
I next consider the MSE-minimizing approach in the case of dynamic treatment effect.
Let $\beta_{\ell}$ denote a coefficient for the treatment effect in the $\ell$th period after the treatment, and suppose that this coefficient is the parameter of interest.
The goal is to minimize the MSE of the estimator for $\beta_{\ell}$.
Although the basic idea is the same as in the static treatment effect case, an additional step is required in the dynamic setting before changing the length of pre-trends.

The estimator of the conventional event study model (\ref{simple event study}) for the treatment effect in the $\ell$th period after the treatment can be written as follows:

\begin{align}\label{event study estimates}
  \hat{\beta}_{\ell} = \frac{1}{\sum_{i = 1}^{N}1\{D_i = 1\}}\sum_{D_i = 1}(Y_{i, t^*_i + \ell} - Y_{i, t^*_i}) - \frac{1}{\sum_{i = 1}^{N}1\{D_i = 0\}}\sum_{D_i = 0}(Y_{i, t^*_i + \ell} - Y_{i, t^*_i})
\end{align}

\noindent
This means the estimator $\hat{\beta}_{\ell}$ is not affected by the values of $Y_{i, t}$ other than $Y_{i, t^*_i + \ell}, Y_{i, t^*_i}$.
Hence, changing the length of pre-trends does not affect the bias and variance of the estimator $\hat{\beta}_{\ell}$.

Therefore, I propose the following modified model (\ref{modified event study}) when the treatment effect is dynamic.

\begin{equation}\label{modified event study}
  Y_{it} =  \mu_i + \theta_t + \sum_{\ell \geq 0} \beta_{\ell} 1\{t - t^*_i = \ell\} + \varepsilon_{it}
\end{equation}

\noindent
Then, the estimator of this modified model (\ref{modified event study}) for the treatment effect in the $\ell$th period after the treatment can be written as follows:

\begin{align}\label{modified event study estimates}
  \hat{\beta}_{\ell}^m = \frac{1}{\sum_{i = 1}^{N}1\{D_i = 1\}}\sum_{D_i = 1}(Y_{i, t^*_i + \ell} - Y_{i, (t \leq t^*_i)}) - \frac{1}{\sum_{i = 1}^{N}1\{D_i = 0\}}\sum_{D_i = 0}(Y_{i, t^*_i + \ell} - Y_{i, (t \leq t^*_i)})
\end{align}

\noindent

Compared with the conventional event-study model, this modified specification does not include the pre-treatment coefficients $\beta_{\ell}$ for $\ell \leq -2$.
As a result, more observations can be used to estimate the post-treatment coefficients $\beta_{\ell}$, which may reduce the variance relative to the conventional event-study model (\ref{simple event study}). 
However, because the modified specification does not estimate coefficients for pre-treatment periods, the conventional pre-trends test cannot be conducted within this model itself. 
For this reason, the use of the modified specification is justified when the data have already passed the conventional pre-trends test under the conventional event-study model (\ref{simple event study}).

A similar approach can then be applied as in the TWFE case. 
When the pre-treatment period is longer, the estimator tends to have a smaller variance because more observations are used in the estimation. 
However, the estimator may also be more biased, since the model imposes the parallel trends assumption between the treatment and control groups over a longer pre-treatment window. 
By contrast, when the pre-treatment period is shorter, the estimator tends to have a larger variance because fewer observations are available. 
At the same time, the potential bias may be smaller, because the model relies on the parallel trends assumption over a narrower pre-treatment window. 
By minimizing the MSE of the estimator, the optimal length can be selected, yielding the MSE-optimal event-study estimator.

\section{Simulation}
In this section, I conduct two simulations to illustrate the proposed model.
I first illustrate the case of the static treatment effect in Section $4.1$. 
I then explain the dynamic treatment effect case in Section $4.2$.

\subsection{Static Treatment Effect Estimation}

This simulation applies the method in Section 3.2 to simulated data with fundamentally non-parallel trends between the treatment and the control groups. 
The simulated data are generated according to the following setup:
\begin{align*}
  &Y_{it} =  \mu_i + \theta_{it} + \beta D_{it} + \varepsilon_{it} \\
  &\text{where} \\ 
  &\begin{cases}
    \mu_{\text{treatment}} &= 100 \\
    \mu_{\text{control}} &= -100 
  \end{cases} \ , \ \theta_{it} =
  \begin{cases}
    11t \ \text{ if } t < 0 \land i \in \text{control} \\
    10t \ \text{ otherwise }
  \end{cases} , \
  \beta = 100, \ \varepsilon_{it} \sim N(0, 50)
\end{align*}
Both the treatment group and the control group are set to have $50$ individuals and $t \in [-10,10]$. 
The trends of each group are shown in Figure \ref{fig: simulation TWFE raw mean}.

The results are reported in Table \ref{MSEs TWFE}.
The bias tends to increase when the pre-trends extend over a longer period.
At the same time, the variance tends to get smaller.
The MSE is minimized when the pre-trends length is $3$, which is the optimal pre-trends length. 
The MSE-optimal TWFE estimate is $95.449 \ (6.076)$.
The conventional TWFE model does not change the pre-trends length; therefore, the conventional TWFE model estimate is $91.570 \ (4.577)$.
Compared with this estimate, the MSE-optimal estimate is less biased but it has a larger variance.

\begin{table}[H]
  \footnotesize
  \centering
  \caption{List of the estimates by each length of pre-trends}
  \label{MSEs TWFE}
  \resizebox{\linewidth}{!}{%
      
\begin{tabular}[t]{cccccc}
\toprule
Length of Pre-trends & MSE & Bias Squared & Variance & Coefficient & SE\\
\midrule
\cellcolor{white}{0} & \cellcolor{white}{124.367} & \cellcolor{white}{ 0.000} & \cellcolor{white}{124.367} & \cellcolor{white}{100.292} & \cellcolor{white}{11.152}\\
\cellcolor{white}{1} & \cellcolor{white}{ 77.355} & \cellcolor{white}{ 6.675} & \cellcolor{white}{ 70.681} & \cellcolor{white}{102.876} & \cellcolor{white}{ 8.407}\\
\cellcolor{white}{2} & \cellcolor{white}{ 68.960} & \cellcolor{white}{15.870} & \cellcolor{white}{ 53.091} & \cellcolor{white}{ 96.309} & \cellcolor{white}{ 7.286}\\
\cellcolor{light-gray}{3} & \cellcolor{light-gray}{ 60.372} & \cellcolor{light-gray}{23.456} & \cellcolor{light-gray}{ 36.916} & \cellcolor{light-gray}{ 95.449} & \cellcolor{light-gray}{ 6.076}\\
\cellcolor{white}{4} & \cellcolor{white}{ 81.921} & \cellcolor{white}{48.350} & \cellcolor{white}{ 33.571} & \cellcolor{white}{ 93.339} & \cellcolor{white}{ 5.794}\\
\addlinespace
\cellcolor{white}{5} & \cellcolor{white}{ 76.763} & \cellcolor{white}{44.352} & \cellcolor{white}{ 32.411} & \cellcolor{white}{ 93.633} & \cellcolor{white}{ 5.693}\\
\cellcolor{white}{6} & \cellcolor{white}{ 84.406} & \cellcolor{white}{56.812} & \cellcolor{white}{ 27.594} & \cellcolor{white}{ 92.755} & \cellcolor{white}{ 5.253}\\
\cellcolor{white}{7} & \cellcolor{white}{108.727} & \cellcolor{white}{84.194} & \cellcolor{white}{ 24.534} & \cellcolor{white}{ 91.117} & \cellcolor{white}{ 4.953}\\
\cellcolor{white}{8} & \cellcolor{white}{115.393} & \cellcolor{white}{92.359} & \cellcolor{white}{ 23.033} & \cellcolor{white}{ 90.682} & \cellcolor{white}{ 4.799}\\
\cellcolor{white}{9} & \cellcolor{white}{ 97.023} & \cellcolor{white}{76.076} & \cellcolor{white}{ 20.947} & \cellcolor{white}{ 91.570} & \cellcolor{white}{ 4.577}\\
\bottomrule
\end{tabular}
 
  }
  \begin{minipage}{14cm}
    \vspace{0.2cm}
    \footnotesize
    Notes: Each row represents the (estimated) values of MSE, bias squared, variance, coefficient, and SE with each length of pre-trends.
    The MSE is minimized when the pre-trends length is $3$ (the shaded row). 
  \end{minipage}
\end{table}

\subsection{Dynamic Treatment Effect Estimation}
The next simulation considers a dynamic treatment effect.
This simulation applies the method in Section 3.3 to different simulated data which also have fundamentally non-parallel trends between the treatment and the control groups. 
The simulated data are generated according to the following setup:
\begin{align*}
  &Y_{it} =  \mu_i + \theta_{it} + \sum_{l \neq -1} \beta_l 1\{t - t^*_i = l\} + \varepsilon_{it} \\
  &\text{where} \\ 
  &\begin{cases}
    \mu_{\text{treatment}} &= 100 \\
    \mu_{\text{control}} &= -100 
  \end{cases} \ , \ \theta_{it} = 
  \begin{cases}
    7t \ \text{ if } t < 0 \land i \in \text{control} \\
    10t \ \text{ otherwise }
  \end{cases} , \
  \beta_l = 6l, \ \varepsilon_{it} \sim N(0, 50)
\end{align*}
In the same way as the simulation in Section $4.1$, both the treatment group and the control group are set to have $50$ individuals and $t \in [-10,10]$. 
The trends of each group are shown in Figure \ref{fig: simulation event study raw mean}.

Suppose that the target parameter is $\beta_5$.
Note that the model to estimate is model (\ref{modified event study}).
The results are reported in Table \ref{MSEs event study beta 5}.
The bias of the estimates changes nonmonotonically but it seems to get larger when the pre-trends extend over a longer period.
The variance tends to get smaller as the pre-trends become longer.
The MSE is minimized when the pre-trends length is $6$, which is the optimal pre-trends length. 
The MSE-optimal event study estimate is $50.943 \ (22.021)$.

\begin{table}[H]
  \footnotesize
  \centering
  \caption{List of the estimates by each length of pre-trends \\ Target parameter = $\beta_{5}$}
  \label{MSEs event study beta 5}
  \resizebox{\linewidth}{!}{%
      
\begin{tabular}[t]{cccccc}
\toprule
Length of Pre-trends & MSE & Bias Squared & Variance & Coefficient & SE\\
\midrule
\cellcolor{white}{0} & \cellcolor{white}{1017.981} & \cellcolor{white}{  0.000} & \cellcolor{white}{1017.981} & \cellcolor{white}{47.099} & \cellcolor{white}{31.906}\\
\cellcolor{white}{1} & \cellcolor{white}{ 676.540} & \cellcolor{white}{ 51.938} & \cellcolor{white}{ 624.602} & \cellcolor{white}{54.306} & \cellcolor{white}{24.992}\\
\cellcolor{white}{2} & \cellcolor{white}{ 645.241} & \cellcolor{white}{ 89.283} & \cellcolor{white}{ 555.958} & \cellcolor{white}{56.548} & \cellcolor{white}{23.579}\\
\cellcolor{white}{3} & \cellcolor{white}{ 508.401} & \cellcolor{white}{ 18.735} & \cellcolor{white}{ 489.665} & \cellcolor{white}{51.428} & \cellcolor{white}{22.128}\\
\cellcolor{white}{4} & \cellcolor{white}{ 525.972} & \cellcolor{white}{  7.822} & \cellcolor{white}{ 518.150} & \cellcolor{white}{49.896} & \cellcolor{white}{22.763}\\
\addlinespace
\cellcolor{white}{5} & \cellcolor{white}{ 504.936} & \cellcolor{white}{ 25.067} & \cellcolor{white}{ 479.870} & \cellcolor{white}{52.106} & \cellcolor{white}{21.906}\\
\cellcolor{light-gray}{6} & \cellcolor{light-gray}{ 499.688} & \cellcolor{light-gray}{ 14.774} & \cellcolor{light-gray}{ 484.914} & \cellcolor{light-gray}{50.943} & \cellcolor{light-gray}{22.021}\\
\cellcolor{white}{7} & \cellcolor{white}{ 526.362} & \cellcolor{white}{ 60.236} & \cellcolor{white}{ 466.126} & \cellcolor{white}{54.860} & \cellcolor{white}{21.590}\\
\cellcolor{white}{8} & \cellcolor{white}{ 600.299} & \cellcolor{white}{133.914} & \cellcolor{white}{ 466.385} & \cellcolor{white}{58.671} & \cellcolor{white}{21.596}\\
\cellcolor{white}{9} & \cellcolor{white}{ 727.893} & \cellcolor{white}{269.806} & \cellcolor{white}{ 458.088} & \cellcolor{white}{63.525} & \cellcolor{white}{21.403}\\
\bottomrule
\end{tabular}
 
  }
  \begin{minipage}{14cm}
    \vspace{0.2cm}
    \footnotesize
    Notes: Each row represents the (estimated) values of MSE, bias squared, variance, coefficient, and SE with each length of pre-trends.
    The MSE is minimized when the pre-trends length is $6$ (the shaded row).
  \end{minipage}
\end{table}

\newpage
\begin{landscape}
  \begin{figure}
    \centering
    \includegraphics[width=0.90\columnwidth]{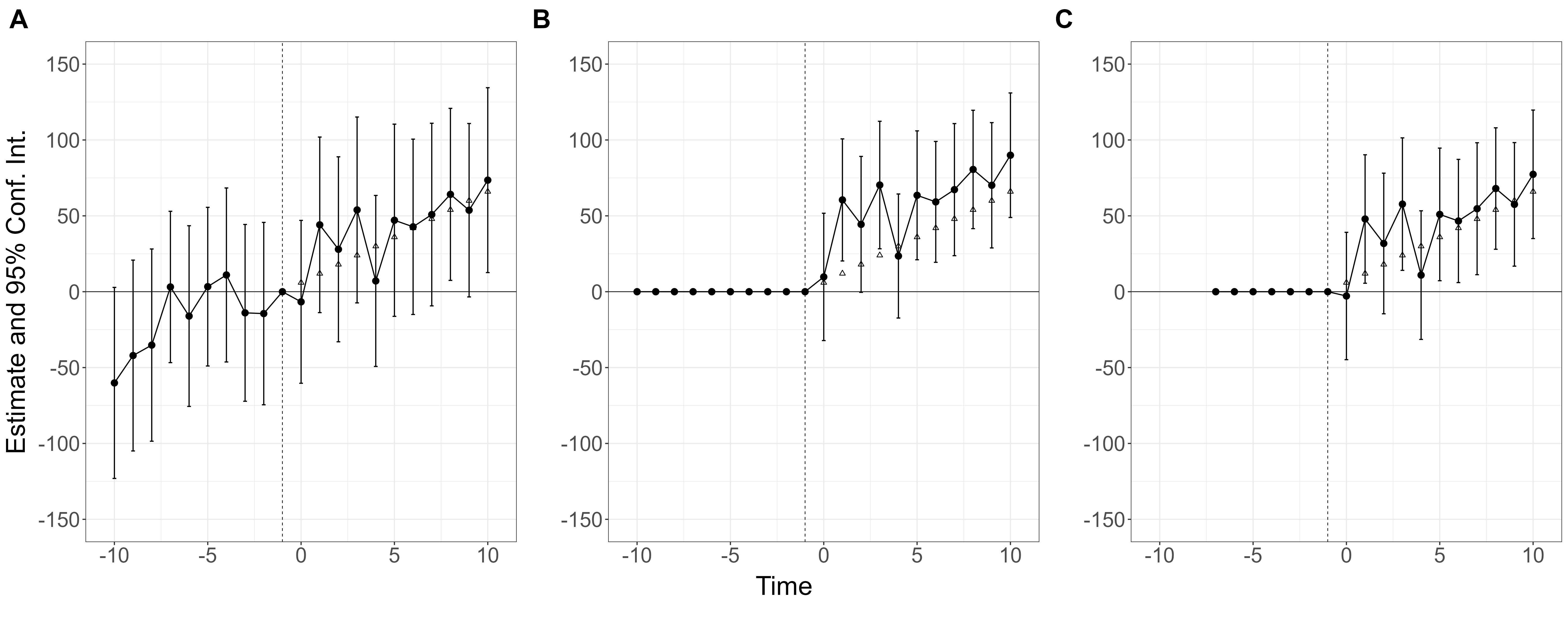}
    \caption{Event study plots with different models}
    \label{fig: simulation event study optimal}
    \begin{minipage}{14cm}
      \footnotesize
      Notes: 
      These figures represent the event study plots with three models: (A) the traditional event study model (\ref{simple event study}),
      (B) the modified event study model (\ref{modified event study}) with the full length of pre-trends,
      (C) the MSE-optimal event study model. 
      The black points are the point estimates and the vertical lines are the $95\%$ confidence intervals.
      The triangle points are the true values of the simulation.
    \end{minipage}
  \end{figure}
\end{landscape}
\newpage

I compare the results with the conventional event study model. 
Figure \ref{fig: simulation event study optimal} (A) represents the event study plot of the conventional event study model (\ref{simple event study}).
It barely passes the pre-trends test, but the estimates have larger variances.
Next, I apply the modified model (\ref{modified event study}).
Figure \ref{fig: simulation event study optimal} (B) shows the event study plot of the modified event study model (\ref{modified event study}) with the full length of pre-trends.
Owing to the increase of the sample size for the estimation, the estimates have smaller variances.
However, imposing the parallel trends assumption on pre-treatment periods for the two groups introduces bias into the estimates.
Lastly, Figure \ref{fig: simulation event study optimal} (C) shows the event study plot of the MSE-optimal event study model. 
The estimates have smaller variances compared with the traditional event study model (\ref{simple event study}) and
are less biased compared with the modified event study model (\ref{modified event study}) with the full length of pre-trends.

\section{Empirical Application}
I apply the MSE-optimal DiD model to the setting of \citeA{Fitzpatrick_and_Lovenheim:2014}.
This research studies the effect of the Early Retirement Incentives (ERI) program in Illinois in 1993 on student achievement.
They estimate regressions of the following form:
\begin{equation}\label{Empirical_application_eq}
  \begin{split}
    Y_{igt}^s =& \beta_0 + \beta_1 (Teachers \geq 15)_{ig} \times Post_t + \beta_2 \textit{Teachers}_{ig} \times Post_t \\ 
    &+ \gamma \mathbb{X}_{it} + \delta_{ig} + \phi_{tg} + \varepsilon_{igt}^s
  \end{split}
\end{equation}
where $Y_{igt}^s$ is the standardized test score in grade $g$ for subject $s$ in school $i$ and year $t$.
The variable $Teachers \geq 15$ is the average number of teachers in a given grade with at least $15$ years of experience pre-$1994$. 
$Teachers$ is the average total number of teachers in a grade and school in the pre-ERI period, and
$Post$ is an indicator variable equal to $1$ for school years after $1993$. 
The vector $\mathbb{X}$ contains the set of school-by-year demographic variables.
$\delta_{ig}$ is school-by-grade fixed effects, and $\phi_{tg}$ is grade-by-year fixed effects.

In one analysis, they estimate the effect of the ERI program on the Math test score using a conventional event study model based on the regression (\ref{Empirical_application_eq}).
The results are shown in the left panel of Figure \ref{fig: F3PA_MSE}.
It passes the pre-trends test but the confidence intervals are relatively wide.
In other words, the variances of the estimates are large.

Suppose that the target parameter is the $1997$ coefficient because the effect of the ERI program on student achievement may arise with a delay.
Then, the results of the MSE-optimal event study model are reported in Table \ref{table: F3PA_MSE}.
Based on the calculation, MSE is minimized when the length of the pre-trends is $3$.
Therefore, the MSE-optimal event study plot is shown in the right panel of the Figure \ref{fig: F3PA_MSE}.

\begin{figure}[H]
  \centering
  \includegraphics[width=0.95\columnwidth]{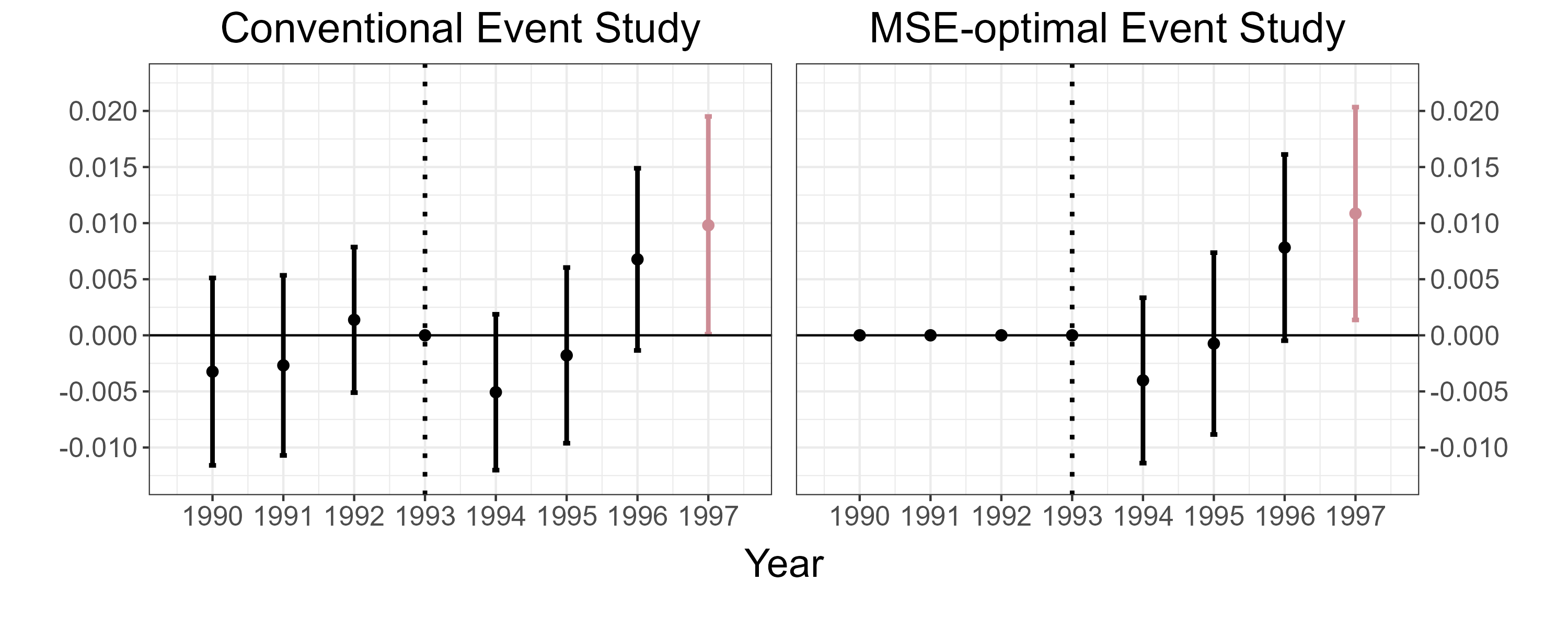}
  \caption{Event Study Estimates of the Effect of the Early Retirement Incentives (ERI) Program on Student Math Test Scores}
  \label{fig: F3PA_MSE}
  \begin{minipage}{14cm}
    \footnotesize
    Notes: This figure explains the effect of ERI program on math test scores.
    It shows both point estimates and 95$\%$ confidence intervals.
    The left plot, based on the conventional event study model, is the equivalent to the plot in \citeA{Fitzpatrick_and_Lovenheim:2014}.
    The right plot is based on MSE-optimal event study model.
    The target parameter is the 1997 coefficient and the MSE of this parameter is minimized. 
  \end{minipage}
\end{figure}

\begin{table}[H]
  \footnotesize
  \centering
  \caption{List of the estimates by each length of pre-trends}
  \label{table: F3PA_MSE}
  \resizebox{\linewidth}{!}{%
      
\begin{tabular}[t]{cccccc}
\toprule
Length of Pre-trends & MSE & Bias Squared & Variance & Coefficient & SE\\
\midrule
\cellcolor{white}{0} & \cellcolor{white}{0.0000250} & \cellcolor{white}{0.0000000} & \cellcolor{white}{0.0000250} & \cellcolor{white}{0.00982} & \cellcolor{white}{0.00500}\\
\cellcolor{white}{1} & \cellcolor{white}{0.0000266} & \cellcolor{white}{0.0000004} & \cellcolor{white}{0.0000262} & \cellcolor{white}{0.00917} & \cellcolor{white}{0.00512}\\
\cellcolor{white}{2} & \cellcolor{white}{0.0000247} & \cellcolor{white}{0.0000004} & \cellcolor{white}{0.0000243} & \cellcolor{white}{0.01042} & \cellcolor{white}{0.00493}\\
\cellcolor{light-gray}{3} & \cellcolor{light-gray}{0.0000245} & \cellcolor{light-gray}{0.0000011} & \cellcolor{light-gray}{0.0000234} & \cellcolor{light-gray}{0.01086} & \cellcolor{light-gray}{0.00484}\\
\bottomrule
\end{tabular}
 
  }
  \begin{minipage}{14cm}
    \vspace{0.2cm}
    \footnotesize
    Notes: This table shows the result of the MSE-optimal event study model.
    The target parameter is the 1997 coefficient.
    Each row represents the (estimated) values of MSE, bias squared, variance, coefficient, and SE with each length of pre-trends.
    The MSE is minimized when the pre-trends length is $3$ (the shaded row).
  \end{minipage}  
\end{table}

\section{Conclusion}
This paper proposes an MSE-optimal approach to difference-in-differences estimation.
Conventional DiD designs (TWFE model, event study model) may suffer from either high variance or substantial bias.
The proposed model changes the length of pre-trends and finds the optimal one that minimizes the MSE of the estimator.
I illustrate the method using simulations and an empirical application.
As future work, it might be possible to apply this method to a staggered treatment setting.

\nocite{*}
\bibliographystyle{apacite}
\bibliography{MSE_optimal_DiD_reference}

@article{Baker_etal:2026,
  title={Difference-in-Differences Designs: A Practitioner's Guide},
  author={Baker, A. C. and Callaway, B. and Cunningham, S. and Goodman-Bacon, A. and Sant'Anna, P. H.},
  journal={Journal of Economic Literature},
  volume={},
  number={},
  pages={},
  year={forthcoming}
}

@article{Baker_etal:2022,
  title={How much should we trust staggered difference-in-differences estimates?},
  author={Baker, A. C. and Larcker, D. F. and Wang, C. C},
  journal={Journal of Financial Economics},
  volume={144},
  number={2},
  pages={370-395},
  year={2022}
}

@misc{Butts:2023,
      title={Difference-in-Differences Estimation with Spatial Spillovers}, 
      author={Butts, K.},
      year={2023},
      eprint={2105.03737},
      archivePrefix={arXiv},
      primaryClass={econ.EM},
      url={https://arxiv.org/abs/2105.03737}, 
}

@article{Callaway_and_Sant'Anna:2022,
  title={Difference-in-differences with multiple time periods},
  author={Callaway, B. and Sant'Anna, P. H.},
  journal={Journal of Econometrics},
  volume={225},
  number={2},
  pages={200-230},
  year={2021}
}

@article{Calonico_and_Cattaneo_and_Titiunik:2014,
  title = {Robust Nonparametric Confidence Intervals for Regression-Discontinuity Designs},
  author = {Calonico, S. and Cattaneo, M. D. and Titiunik, R.},
  journal = {Econometrica},
  volume = {82},
  number = {6},
  pages = {2295-2326},
  year = {2014}
}

@article{Card_and_Krueger:2022,
  title={Minimum wages and employment: A case study of the fast-food industry in New Jersey and Pennsylvania},
  author={Card, D. and Krueger, A. B.},
  journal={American Economic Review},
  volume={84},
  number={4},
  pages={772-793},
  year={1994}
}

@article{de_Chaisemartin_and_d'Haultfoeuille:2020,
  title={Two-way fixed effects estimators with heterogeneous treatment effects},
  author={de Chaisemartin, C. and d'Haultfoeuille, X.},
  journal={American Economic Review},
  volume={110},
  number={9},
  pages={2964-2996},
  year={2020}
}

@article{de_Chaisemartin_and_d'Haultfoeuille:2023,
  title={Two-way fixed effects and differences-in-differences with heterogeneous treatment effects: A survey},
  author={de Chaisemartin, C. and d'Haultfoeuille, X.},
  journal={The Econometrics Journal},
  volume={26},
  number={3},
  pages={C1-C30},
  year={2023}
}

@article{Egami_and_Yamauchi:2023,
  title={Using multiple pretreatment periods to improve difference-in-differences and staggered adoption designs},
  author={Egami, N. and Yamauchi, S.},
  journal={Political Analysis},
  volume={31},
  number={2},
  pages={195-212},
  year={2023}
}

@misc{Fiorini_Lee_and_Pfeifer:2024,
  title={A simple approach to staggered difference-in-differences in the presence of spillovers},
  author={Fiorini, M. and Lee, W. and Pfeifer, G.},
  publisher={CESifo Working Paper No.11011},
  url={https://ssrn.com/abstract=4772584},
  year={2024}
}

@article{Fitzpatrick_and_Lovenheim:2014,
  title={Early retirement incentives and student achievement},
  author={Fitzpatrick, M. D. and Lovenheim, M. F.},
  journal={American Economic Journal: Economic Policy},
  volume={6},
  number={3},
  pages={120-154},
  year={2014}
}

@misc{Fitzpatrick_and_Lovenheim:2019,
  title={Replication data for: Early Retirement Incentives and Student Achievement},
  author={Fitzpatrick, M. D. and Lovenheim, M. F.},
  publisher={American Economic Association [publisher], Inter-university Consortium for Political and Social Research [distributor]},
  url={https://doi.org/10.3886/E114865V1},
  year={2019}
}

@article{Goodman-Bacon:2021,
  title={Difference-in-differences with variation in treatment timing},
  author={Goodman-Bacon, A.},
  journal={Journal of Econometrics},
  volume={225},
  number={2},
  pages={254-277},
  year={2021}
}

@book{Hansen:2022,
  title={Econometrics},
  author={Hansen, B.},
  publisher={Princeton University Press},
  year={2022}
}

@misc{Ishimaru:2026,
      title={Estimating Treatment Effects in Panel Data Without Parallel Trends}, 
      author={Ishimaru, S.},
      year={2026},
      archivePrefix={arXiv},
      url={https://arxiv.org/abs/2601.08281}, 
}

@article{Imbens_and_Kalyanaraman:2012,
  title={Optimal bandwidth choice for the regression discontinuity estimator},
  author={Imbens, G. and Kalyanaraman, K.},
  journal={The Review of Economic Studies},
  volume={79},
  number={3},
  pages={933-959},
  year={2012}
}

@article{Rambachan_and_Roth:2023,
  title={A more credible approach to parallel trends},
  author={Rambachan, A. and Roth, J.},
  journal={The Review of Economic Studies},
  volume={90},
  number={5},
  pages={2555-2591},
  year={2023}
}

@article{Roth:2022,
  title={Pretest with caution: Event-study estimates after testing for parallel trends},
  author={Roth, J.},
  journal={American Economic Review: Insights},
  volume={4},
  number={3},
  pages={305-322},
  year={2022}
}

@article{Roth_etal:2023,
  title={What's trending in difference-in-differences? A synthesis of the recent econometrics literature},
  author={Roth, J. and Sant'Anna, P. H. and Bilinski, A. and Poe, J.},
  journal={Journal of Econometrics},
  volume={235},
  number={2},
  pages={2218-2244},
  year={2023}
}

@article{Sun_and_Abraham:2021,
  title={Estimating dynamic treatment effects in event studies with heterogeneous treatment effects},
  author={Sun, L. and Abraham, S.},
  journal={Journal of Econometrics},
  volume={225},
  number={2},
  pages={175-199},
  year={2021}
}

\appendix
\renewcommand{\thetable}{ \Alph{section}.\arabic{table}}
\renewcommand{\thefigure}{ \Alph{section}.\arabic{figure}}
\setcounter{table}{0}    
\setcounter{figure}{0}

\section{Figures for Baseline DiD Models}
In Section $2.1$, I formulate the simplest $2 \times 2$ DiD setting, and Figure \ref{fig: DiD} illustrates how the ATT is identified in this framework. 
In Sections $2.2$ and $2.3$, I introduce two commonly used DiD specifications: the TWFE model and the event-study model. 
The TWFE specification is appropriate when the treatment effect is static, whereas the event-study specification is appropriate when the treatment effect is dynamic. 
Figure \ref{fig: canonical plots} illustrates the difference between these two cases.

\begin{figure}[H]
  \centering
  \includegraphics[width=0.50\columnwidth]{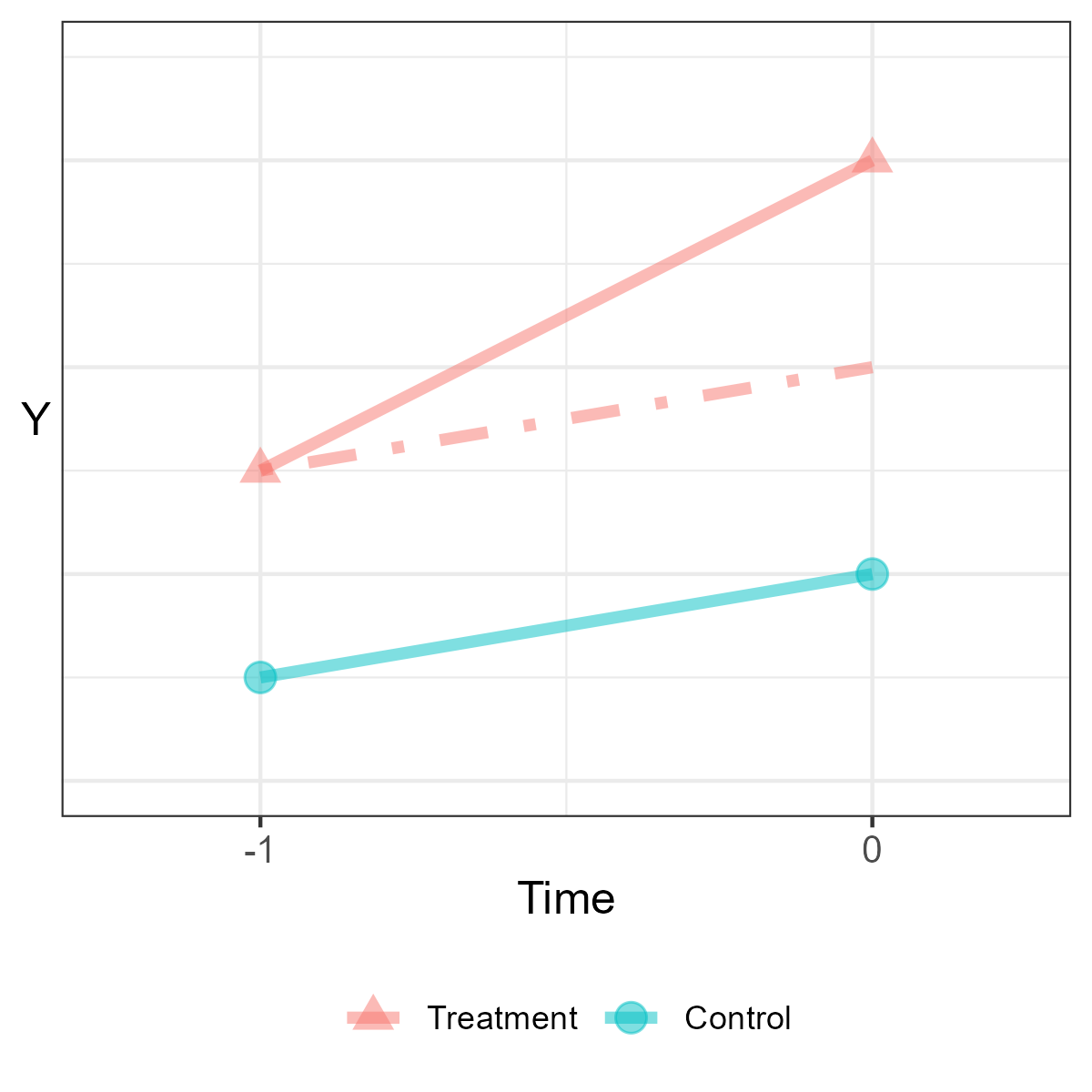}
  \caption{$2 \times 2$ DiD Design}
  \label{fig: DiD}
  \begin{minipage}{14cm}
    \footnotesize
    Notes: This figure overviews the changes of outcomes with both the treatment group (red line) and the control group (blue line). 
    The dashed red line shows the counterfactual outcome of the treatment group without a treatment.
    The difference between the red line and the red dashed line represents the ATT. 
  \end{minipage}
\end{figure}

\begin{figure}[H]
  \centering
  \includegraphics[width=0.85\columnwidth]{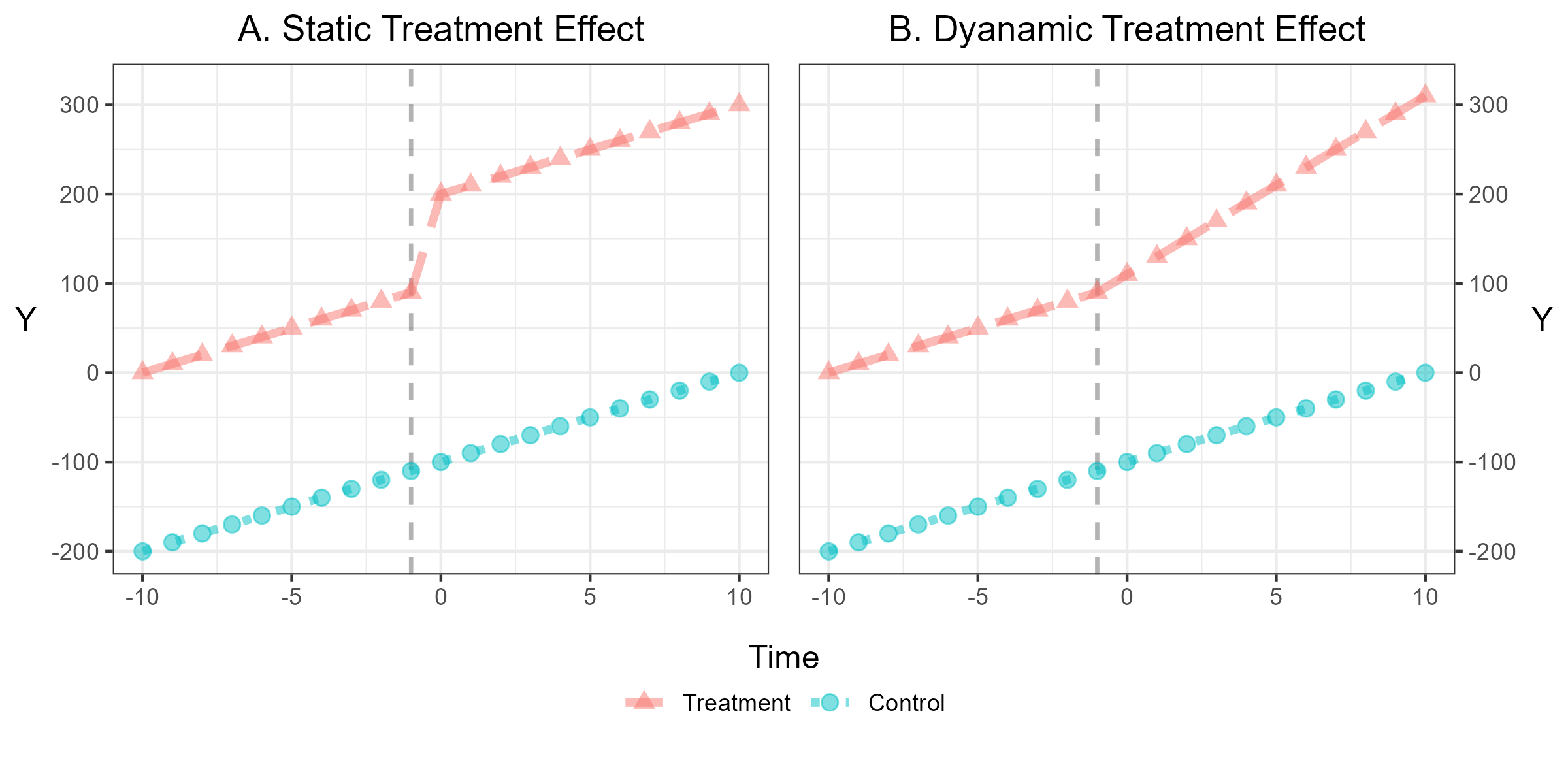}
  \caption{DiD with static and dynamic treatment effects}
  \label{fig: canonical plots}
  \begin{minipage}{14cm}
    \footnotesize
    Notes: The figures explain the changes of outcomes with both the treatment group (red line) and the control group (blue line). 
    In the left figure, the treatment effect is constant across post-treatment periods, that is, the treatment effect is static. 
    In the right figure, the treatment effect varies across post-treatment periods, that is, the treatment effect is dynamic.
  \end{minipage}
\end{figure}

\newpage

\section{Additional Figures of Simulation}
The following two figures are the overviews of the simulation data in Section $4.1$ and $4.2$.
\begin{figure}[H]
  \centering
  \includegraphics[width=0.70\columnwidth]{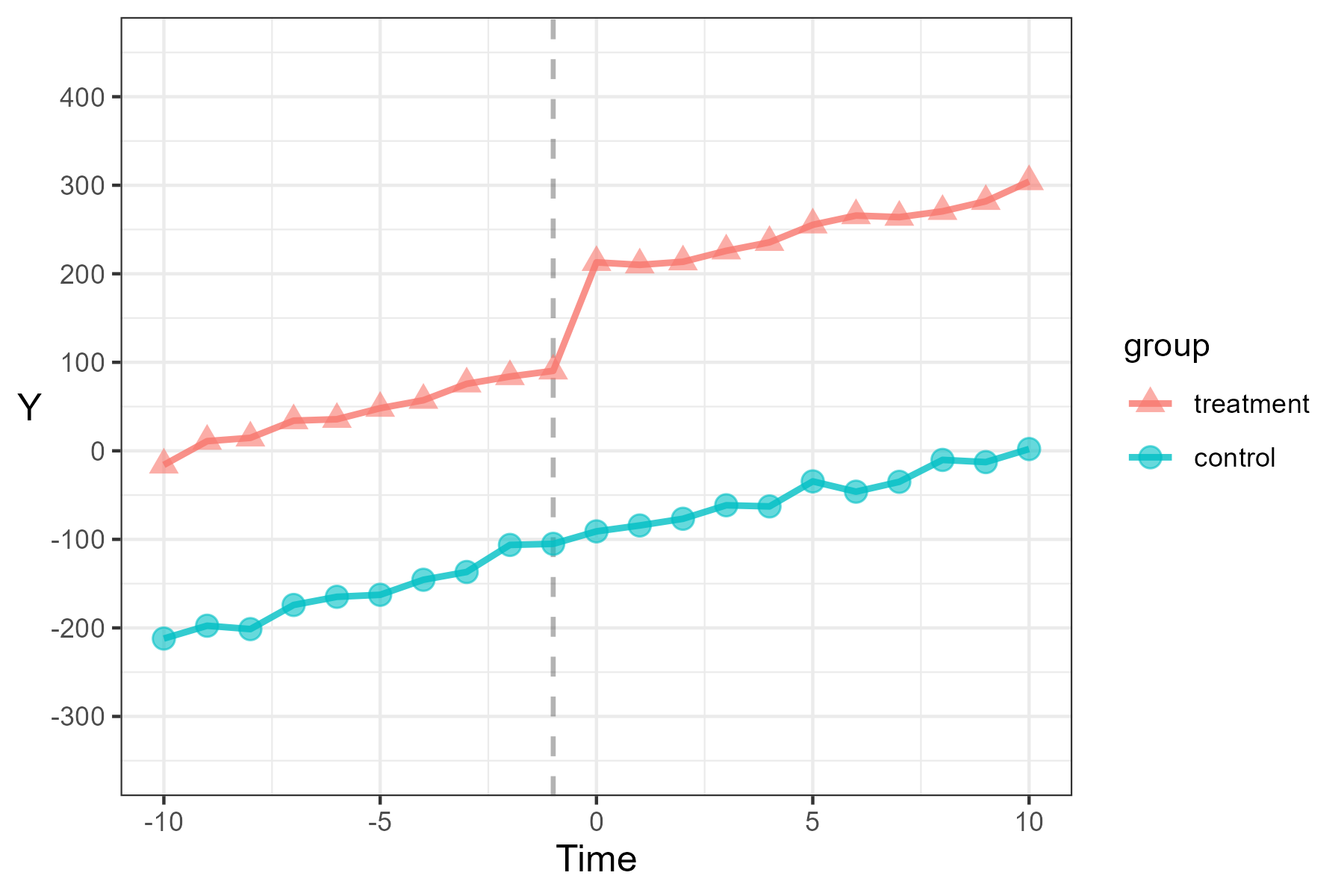}
  \caption{Trends of the 2 groups with static treatment effect}
  \label{fig: simulation TWFE raw mean}
  \begin{minipage}{14cm}
    \footnotesize
    Note: This figure represents the trends of the mean values of the two groups: the treatment group (red line) and the control group (blue line).
    The treatment effect is static, i.e., time-invariant.
  \end{minipage}
\end{figure}

\begin{figure}[H]
  \centering
  \includegraphics[width=0.70\columnwidth]{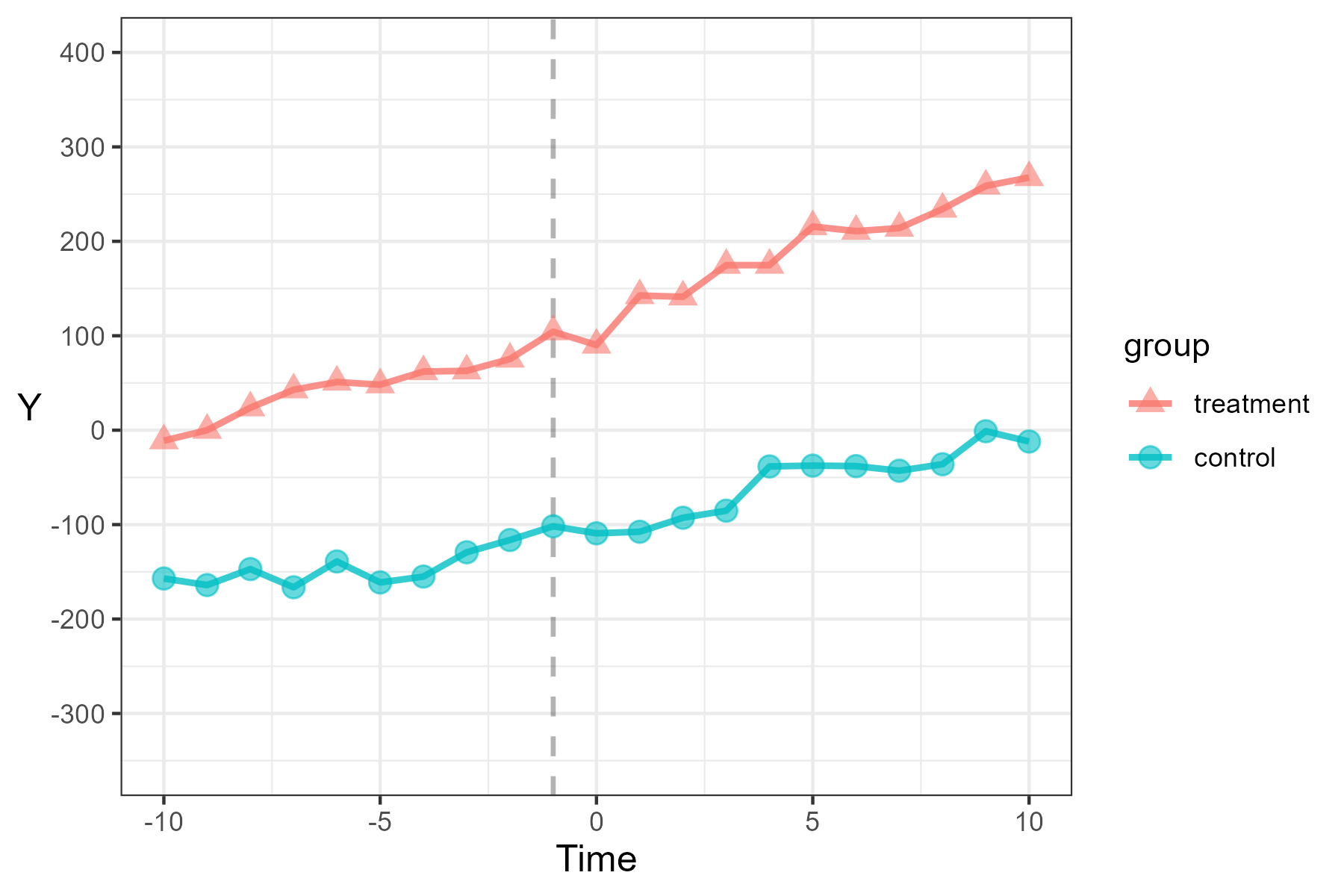}
  \caption{Trends of the 2 groups with dynamic treatment effect}
  \label{fig: simulation event study raw mean}
  \begin{minipage}{14cm}
    \footnotesize
    Note: This figure represents the trends of the mean values of the two groups: the treatment group (red line) and the control group (blue line).
    The treatment effect is dynamic, i.e., time-variant.
  \end{minipage}
\end{figure}

\end{document}